\documentclass[reqno,11pt]{amsart}
\usepackage{amssymb,amsmath,amsthm,amsfonts,times,graphicx}

\usepackage{color,setspace,wrapfig,multicol}
\usepackage[labelformat=empty]{subcaption}

\usepackage[left=3.5cm,right=3.5cm,top=3cm,bottom=3cm]{geometry}

\usepackage[colorlinks=true, pdfstartview=FitV, linkcolor=blue, 
            citecolor=blue, urlcolor=blue]{hyperref}

\numberwithin{equation}{section} 

\newcommand{\C}{\mathbb{C}}
\newcommand{\R}{\mathbb{R}}
\newcommand{\T}{\mathbb{T}}
\newcommand{\Z}{\mathbb{Z}}
\renewcommand{\H}{\mathbb{H}}
\newcommand{\V}{\mathbb{V}}

\newcommand{\cE}{\mathcal{E}}
\newcommand{\cF}{\mathcal{F}}

\newcommand{\cS}{\textit{S}}
\newcommand{\cC}{\textit{C}}

\renewcommand{\epsilon}{\varepsilon}

\renewcommand{\d}{\mathrm{d}}
\renewcommand{\hat}{\widehat }

\newcommand{\ess}{\mathrm{ess}}

\newcommand{\miqdor}{\mathrm{n}}

{\bf}{\it}
\newtheorem{theorem}{Theorem}[section]
\newtheorem{lemma}[theorem]{Lemma}

\newtheorem{proposition}[theorem]{Proposition}
\newtheorem{remark}[theorem]{Remark}

\newcommand{\blue}{\color{blue}}
\newcommand{\black}{\color{black}}

\date{\today}

\date{\today}
\begin{document}


\title[Bose-Hubbard model]{Bose-Hubbard models with on-site and nearest-neighbor interactions: Exactly solvable case}

\author{Saidakhmat Lakaev, Shokhrukh Kholmatov,
Shakhobiddin Khamidov}

\address[S. Lakaev]{Samarkand State University, University boulevard, 15, 140104 Samarkand, Uzbekistan}
\email{slakaev@mail.ru}

\address[Sh. Kholmatov]{University of Vienna,
Oskar-Morgenstern-Platz 1, 1090  Vienna, Austria}
\email{shokhrukh.kholmatov@univie.ac.at}

\address[Sh. Khamidov]{Samarkand State University, University boulevard, 15, 140104 Samarkand, Uzbekistan}
\email{shoh.hamidov1990@mail.ru}

\begin{abstract}
We study the discrete spectrum of the two-particle Schr\"odinger
operator $\hat H_{\mu\lambda}(K),$ $K\in\T^2,$ associated to the
Bose-Hubbard Hamiltonian $\hat  \H_{\mu\lambda}$ of a system of two identical bosons interacting on site and nearest-neighbor sites in the two dimensional lattice $\Z^2$  with interaction magnitudes
$\mu\in\R$ and $\lambda\in\R,$ respectively. We completely describe
the spectrum of $\hat H_{\mu\lambda}(0)$ and establish the optimal lower
bound for the number of eigenvalues of $\hat H_{\mu\lambda}(K)$ outside
its essential spectrum for all values of $K\in\T^2.$ Namely, we
partition the $(\mu,\lambda)$-plane such that in each connected
component of the partition the number of bound states of
$\hat H_{\mu\lambda}(K)$ below or above its essential spectrum cannot be
less than the corresponding number of bound states of
$\hat H_{\mu\lambda}(0)$ below or above its essential spectrum.
\end{abstract}


\maketitle

\section{Introduction}

In this paper we consider the family $\hat H_{\mu\lambda}(K)$ Schr\"odinger operators, associated to the Bose-Hubbard Hamiltonian
of a system of two identical bosons on the two dimensional lattice
$\Z^2$ with on-site interaction $\mu\in\R$ and nearest-neighbor
interaction $\lambda\in\R.$

Lattice Bose-Hubbard models have become popular in recent years
since they represent a minimal, natural Hamiltonian in ultracold
atoms in optical lattices, systems with highly controllable
parameters  such as lattice geometry and dimensionality, particle
masses, tunneling, two-body potentials, temperature etc.  (see e.g.,
\cite{B:2005_nat.phys,JBC:1998_phy.rev,JZ:2005_ann.phys,LSA:2012_book}
and references therein). Unlike the traditional condensed matter
systems, where stable composite objects are usually formed by
attractive forces and repulsive forces separate particles in free
space, the controllability of collision properties of ultracold
atoms has enabled to experimentally observe a stable
\emph{repulsive} bound pair of ultracold atoms in the optical
lattice $\Z^3$, see e.g.,
\cite{HWCR:2012,OOHE:2006,WTLG:2006_nature,ZNON:2008}. In all these
observations Bose-Hubbard Hamiltonians became a link between the
theoretical basis and experimental results.

The main difficulty in solving Bose-Hubbard Hamiltonian even with
the on-site interaction is due to the tunneling, i.e., the kinetic
energy necessary for a boson to hop from site to site, since it is
highly nonlocal. Moreover, unlike its continuous counterpart,  the
lattice Hamiltonian, corresponding to short-range interacting sytems
of particle pairs, is not rotationally invariant, hence, the
separation of lattice Hamiltonian related to the center of
mass-motion is not possible. However, the translation-invariance of
the Hamiltonian (in $d$-dimensional lattice $\Z^d$) allows to use
the Floquet-Bloch decomposition: the underlying Hilbert space
$\ell^2 (\Z^d)^2$ and the total two-particle Hamiltonian $\hat
\H$ are represented as the von Neumann direct integral
\begin{equation*}
\ell^2[(\Z^d)^2] \simeq \int\limits_{K\in \T^d} \oplus
\ell^2(\Z^d)\d K, \qquad \hat \H\simeq\int\limits_{K\in \T^d} \oplus
\hat H(K)\d K,
\end{equation*}
associated to the representation of the discrete group $\Z^d$ by the
shift operators, here $\T^d$ is the $d$-dimensional torus, a
Pontryagin dual group to $\Z^d.$  The fiber Hamiltonians $\hat
H_{\mu}(K)$ nontrivially depend on the so-called {\it
quasi-momentum} $K\in\T^d$ (see e.g.,
\cite{ALMM:2006_cmp,FIC:2002_phys.rev,M:1986_rev:phys,M:1991_adv.sov}).

In what follows any nonzero solution $\hat \psi_K$ to the Schr\"odinger
equation
$$
\hat H(K) \hat \psi_K=e_K \hat \psi_K,  \qquad \hat \psi_K \in \ell^2 (\Z^d),
$$
will be called bound state with energy $e_K$.

In the current paper, we study the family
$$
\hat H_{\mu\lambda}(K):=\hat H_0(K) + \hat V_{\mu\lambda},\qquad
K\in\T^2,
$$
discrete Schr\"odinger operators associated to the Hamiltonian
$\hat\H_{\mu\lambda}$ of a system of two identical bosons on the two
dimensional lattice $\Z^2$ with zero-range on-site interaction
$\mu\neq0$ and nearest-neighbor interaction $\lambda\neq0$ (see \eqref{hash_shapka} in 
Subsection \ref{subsec:von_neuman}). To the best of our knowledge,
this is a new, exactly solvable model, for which the exact number of
eigenvalues and their locations as well as exact lower and upper
bounds for all values of the pair interactions $\mu, \lambda \in \R$
can be found.

First we observe that the essential spectrum of $\hat
H_{\mu\lambda}(K)$ consists of a segment
$[\cE_{\min}(K), \allowbreak \cE_{\max}(K)],$ where
$$
\cE_{\min}(K):= 2\sum\limits_{i=1}^2\Big(1-\cos \tfrac{K_{i}}2\Big),
\qquad \cE_{\max}(K):= 2\sum\limits_{i=1}^2\Big(1+\cos
\tfrac{K_{i}}2\Big)
$$
(see Subsection \ref{subsec:ess_spec}).

Now we study the discrete spectrum of $\hat H_{\mu\lambda}(K).$ The
eigenvalue problem for $\hat H_{\mu\lambda}(K)$ with general $K$ is
not an  easy problem; note that $\hat H_{\mu\lambda}(K)$ becomes a
small perturbation of $\hat V_{\mu\lambda}$ if $K$ is close to $\vec\pi:=(-\pi,\pi),$ and $\hat V_{\mu\lambda}$ is a small
pertubation of $\hat H_0$ provided that $\mu,\lambda$ is small.
However, it turns out that for any $\mu,\lambda\in\R,$ if
$e_n(K;\mu,\lambda)$ and $E_n(K;\mu,\lambda)$ are the $n$-th
eigenvalues of $\hat H_{\mu\lambda}(K)$ below and above the
essential spectrum, respectively, then the functions $\cE_{\min}(K)
- e_n(K;\mu,\lambda) $ and $E_n(K;\mu,\lambda) - \cE_{\max}(K)$ as a
function of $K$ will have a minimum at $K=0$ (see Lemma
\ref{lem:monoton_xos_qiymat}). Subsequently, the number of
eigenvalues $\miqdor_+(\hat H_{\mu\lambda}(K))$ (resp.
$\miqdor_-(\hat H_{\mu\lambda}(K))$) of $\hat H_{\mu\lambda}(K)$
above (resp. below) the essential spectrum is not less than that of
$\hat H_{\mu\lambda}(0)$ (Theorem \ref{teo:disc_Kvs0}). This result is a generalization of \cite[Theorems 1 and 2]{ALMM:2006_cmp} which contain the assertion related only to the ground state of $\hat H_{\mu\lambda}(K).$

To study the discrete spectrum of $\hat H_{\mu\lambda}(0)$ we introduce the Fredholm determinants $\Delta_{\mu\lambda}^{(s)}(z)$ and $\Delta_{\lambda}^{(a)}(z),$ associated to the restriction of $\hat H_{\mu\lambda}(0)$ onto symmetric and antisymmetric functions in $\ell^{2,e}(\Z^2)= \ell^{2,e}(\Z\times\Z)$ w.r.t. coordinate permutations\footnote{Note that the terms 'symmetric' and 'antisymmetric' are not related to the particle symmetry, but rather to the invariant subspaces of $\hat H_{\mu\lambda}(0)$ related to the permutation of coordinates in $\Z^2:=\Z\times\Z$ of one-particle.}, respectively. It is well-known that eigenvalues of $H_{\mu\lambda}(0)$ are zeros of these determinants, and vice versa (see also Lemma \ref{lem:det_zeros_vs_eigen}). Moreover, {\it the number of zeros of $\Delta_{\mu\lambda}^{(s)}(z)$ and $\Delta_{\lambda}^{(a)}(z)$ can  change if and only if their asymptotics  as $z$ approaches to the threholds of the essential spectrum vanish.}

Therefore, we partition $(\mu,\lambda)$-plane of interactions into connected components by means of hyperbolas $\lambda\mu + 4\lambda+2\mu=0$ and $\lambda\mu - 4\lambda-2\mu=0$ and straight lines $\lambda=\pm\frac{\pi}{4-\pi}$ (see Figure \ref{fig:sohalar}), where 
$\lambda\mu \pm 4\lambda\pm2\mu$  and $\lambda=\frac{\pi}{4-\pi}$  appear as constant in front of the main term in the asymptotics of the Fredholm determinants $\Delta_{\mu\lambda}^{(s)}(z)$ and $\Delta_{\lambda}^{(a)}(z),$ respectively, as $z$ converges to the edges of the essential spectrum (Proposition \ref{prop:asymp_determinants}). 
As we noticed above, {\it the number of eigenvalues  of $H_{\mu\lambda}(0)$  changes if and only if one of the 'constants' $\lambda\mu \pm 4\lambda\pm2\mu $ and $\lambda\pm\frac{\pi}{4-\pi}$ changes its sign} (see also Lemmas \ref{lem:antisim_zeros}-\ref{lem:sim_negative_zeros}).


%
\begin{wrapfigure}{L}{0.5\textwidth} 
\centering 
\includegraphics[width=0.48\textwidth]{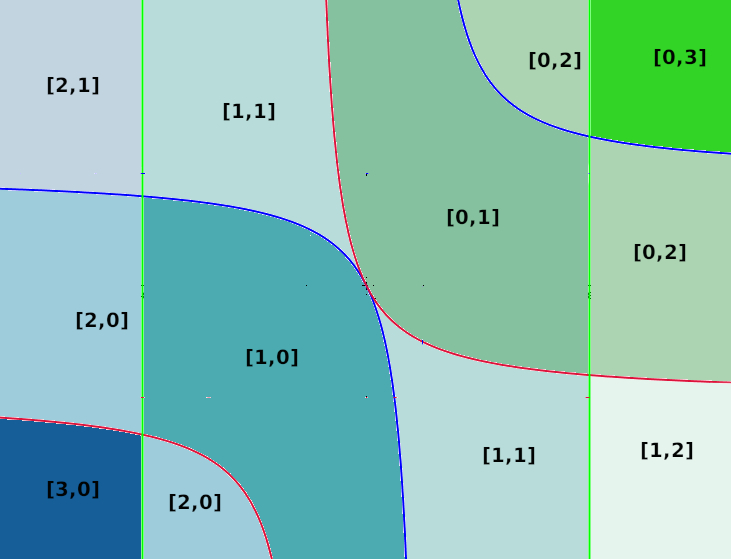} 
\caption{\small A partition of the $(\mu,\lambda)$-plane using the hyperbolas $\lambda\mu\pm 4\lambda\pm2\mu=0$ and straight lines $\lambda=\pm\frac{\pi}{4-\pi}.$ The figure shows also the dynamics of change in the number of eigenvalues: in the region $[i,j]$ the operator $H_{\mu\lambda}(0)$ has exactly $i$ eigenvalues below the essential spectrum and exactly $j$ eigenvalues above the essential spectrum. If $(\mu,\lambda)$ moves from the region $[i,j]$ to the region $[i-1,j]$ resp. $[i,j-1]$ one eigenvalue of $H_{\mu\lambda}(0)$  is absorbed into the essential spetrum at the lower resp. upper edge of the essential spectrum. Conversely, if $(\mu,\lambda)$ moves from the region $[i,j]$ to the region $[i+1,j]$ resp. $[i,j+1]$ one eigenvalue of $H_{\mu\lambda}(0)$  is released the essential spetrum at the lower resp. upper edge of the essential spectrum. This absoption and release of eigenvalues can occur simultaneously: for instance, if $(\mu,\lambda)$ moves from the region $[1,0]$ to $[0,1]$ then one eigenvalue is absorbed at the bottom and another releases at the top.} \label{fig:sohalar}
\end{wrapfigure}
Hence, while $(\mu,\lambda)$ runs in $\R^2$ and does not cross any of these curves, no qualitative or quantitative changes occur in the discrete spectrum of $\hat H_{\mu\lambda}(0);$ however, as soon as $(\mu,\lambda)$ crosses any of those hyperbolas resp. straight lines, the essential spectrum of $\hat H_{\mu\lambda}(0)$ either 'gives birth' or 'absorb' a bound state of $\hat H_{\mu\lambda}(0)$ which is symmetric resp. antisymmetric w.r.t. permutation of coordinates  (Theorem \ref{teo:sharpness}).

This fact is tightly connected to \emph{coupling
constant threshold phenomenon}  \cite{KS:1980_ann.phys,LKh:2012_izv,LKh:2011_jpa}: if $H(t),$ $t\ge0,$ is a one-parameter family of Schr\"odinger operators (in continuum or lattice) and $0$ is a lower edge of the essential spectrum, then $t_0\ge0$ is a coupling constant threshold if and only if there exists a negative eigenvalue $E(t)$ of $H(t)$ for $t>t_0$ such that $E(t)\nearrow0$ as $t\searrow t_0,$ i.e., as
$t\searrow t_0$ an eigenvalue is absorbed at the threshold of continuum, and conversely, as $t\nearrow t_0+\epsilon$ an eigenvalue
is released from the continuum. Indeed, considering $\hat H(t):=\hat H_{t\mu,t\lambda}$ for $t\ge0,$ from Figure \ref{fig:sohalar} below we immediately obtain that the only coupling constant thresholds are $t_0=0,$ $t_0=\big|\frac{4\lambda + 2\mu}{\lambda\mu}\big|$ (when $\lambda\mu\ne0$) and $t_0=\frac{\pi}{4-\pi}.$

Surprisingly, the maximum number of isolated eigenvalues is achieved
only in four connected components in which both $\mu$ and $\lambda$
run on infinite intervals.

In general, some results such as existence of an eigenvalue and also the finiteness of the number of eigenvalues can be obtained for wide classes of operators (see e.g., \cite{KhLA:2020_arxiv,K:1977_ann.phys,KS:1980_ann.phys,S:1976_ann.phys}). However,  Figure \ref{fig:sohalar} shows  that the study of a qualitative  change in the number of eigenvalues of $\hat H_{\mu\lambda}(K),$ even for $K=0,$ is very delicate: there are discs in the $(\mu,\lambda)$-plane with arbitrarily small radius in which the number of eigenvalues has jump (see Theorem \ref{teo:sharpness}).

Recall that in \cite{KhLA:2020_arxiv} authors extend the results of
\cite{K:1977_ann.phys,S:1976_ann.phys} to the lattice case. Namely,
the spectral properties of one-particle discrete Schr\"odinger
operator
$$
\hat {\bf h}_t=\hat {\bf h}_0 + t \hat {\bf v},\qquad t\ge0,
$$
in $\Z^1$ and $\Z^2$ have been studied, here $\hat {\bf h}_0$ is a
self-adjoint Laurent-Toeplitz-type operator generated by a
dispersion relation $\hat \cE : \Z^d\to \C$ of the particle and the
potential $\hat {\bf v}$ is the multiplication operator by $\hat
v:\Z^d\to\R.$ Under certain regularity assumption on $\hat\cE $ and
a decay assumption on $\hat v$, authors show that  if $\sum\hat
v(x)$ is nonnegative resp. nonpositive, then the discrete spectrum
of $\hat{\bf h}_t$ above resp. below its essential spectrum is
non-empty for any $t>0.$ Moreover, authors prove  the existence of
$t_0>0$ (depending only on $\hat v$ and $\hat \cE$) such that
\begin{itemize}
\item[(a)]  if $\sum\hat v(x)>0,$ then for any $t\in(0,t_0)$ the  operator $\hat{\bf h}_t$ has no eigenvalues below the essential spectrum and has a unique eigenvalue above the essential spectrum;

\item[(b)]  if $\sum\hat v(x)<0,$ then for any $t\in(0,t_0)$ the  operator $\hat{\bf h}_t$ has no eigenvalues above the essential spectrum and has a unique eigenvalue below the essential spectrum;

\item[(c)]  if $\sum\hat v(x)=0,$ then for any $t\in(0,t_0)$ the  operator $\hat{\bf h}_t$ has unique eigenvalues in both below and above the essential spectrum.
\end{itemize}
\begin{wrapfigure}{l}{0.5\textwidth}
\centering 
\includegraphics[width=0.49\textwidth]{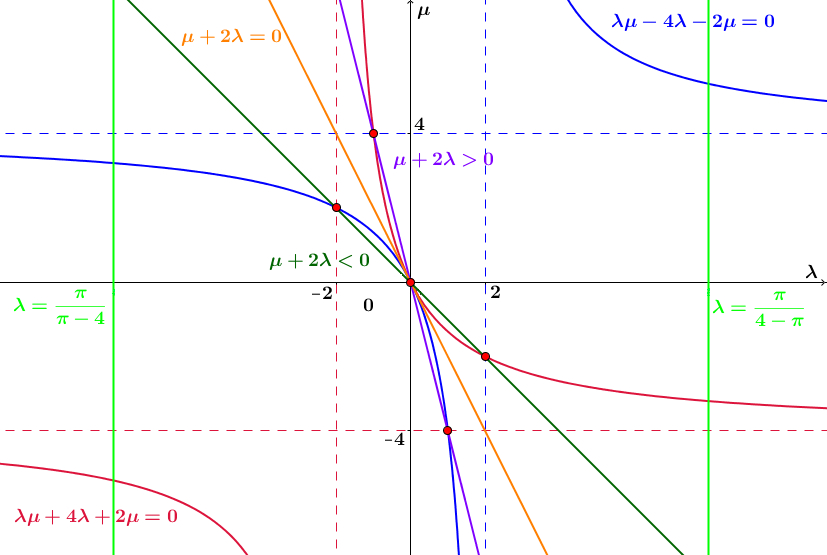} 
\caption{\small The definition of $t_0$ given by \cite{KhLA:2020_arxiv}.} \label{fig:klaus_grafik}
\end{wrapfigure}
Note that for our model $\sum\hat v(x)=\mu+2\lambda.$ Hence, our results for the discrete spectrum of $\hat H_{\mu\lambda}(0)$ show the (exact!) dependence of $t_0$ on $\mu$ and $\lambda:$ namely, $t_0$ will the intersection of the half-line $\{(t\mu,t\lambda):\,\,t\ge0\}$ with the branches of hyperbolas $\lambda\mu \pm 4\lambda\pm2\mu =0$ passing through the origin (see Figure \ref{fig:klaus_grafik}).

Moreover, from Figure \ref{fig:klaus_grafik} we observe that the unique eigenvalue of \cite{KhLA:2020_arxiv} appears from the symmetric space and eigenvalues in the antisymmetric space come out later (see Theorem \ref{teo:KSexistence}).

In general, it is well-known that in the case of $d\ge3$ or in the case of fermions with $d\ge1$ the bound states appear from the essential spectrum either as a threshold bound state or as a threshold resonance \cite{LA:2020_tmf,LB:2009_tmf}. However, our results show that in the case $d=2,$ even though antisymmetric bound states come out from the threshold eigenvalues, all symmetric bound states of $\hat H_{\mu\lambda}(0)$ appear (only!) from the singularity of the associated Fredholm determinant at the thresholds, namely, associated Fredholm determinant cannot be analytically (even continuously) extended to the edges of the essential spectrum (see Proposition \ref{prop:asymp_determinants}). Such a result holds also in $d=1$ (see \cite{LO:2014_2016}). 
This is a strict mathematical explanation in the difference of the appearance of eigenvalues  for $d=1,2$ and $d\ge3,$ or for bosons and fermions.

As far as we know for continuous two-body Schr\"odinger operators $\R^2$ there are no analogous examples and results. In the discrete case similar results for the number of eigenvalues of one-particle Schr\"odinger operators in $\Z^d$ with zero-range on-site and nearest-neighbor interactions have been obtained, for instance, in \cite{LB:2009_tmf} for $d=3$ with attractive potential field, in \cite{LO:2014_2016} for $d=1,$ and in \cite{HMK:2020_lmalg} for all $d\ge1$ considering only negative eigenvalues.

The paper is organized as follows: In Section \ref{sec:hamiltonian}
we introduce the Hamiltonian $\hat\H_{\mu\lambda}$ of a system of two
bosons in the position and momentum space and  also the
Schr\"{o}dinger operator $\hat H_{\mu\lambda}(K)$ associated to
$\hat\H_{\mu\lambda}$. Main results of the paper are stated in Section
\ref{sec:main_results} and their proofs are contained in Section
\ref{sec:proofs}. We conclude the paper with an appendix containing
the proof of Proposition \ref{prop:asymp_determinants}.

\section{Discrete Schr\"odinger operators on lattices} \label{sec:hamiltonian}

\subsection{The two-particle Hamiltonian: the position-space representation}

Let $\Z^2$ be the two dimensional lattice and let
$\ell^{2,s}(\Z^2\times\Z^2)$ be the Hilbert space of square-summable
symmetric functions on $\Z^2\times\Z^2.$

In the position-space representation the two particle
Hamiltonian $\hat \H_{\mu\lambda},$ associated to a system of two
bosons interacting via zero-range and nearest-neighbor potential
$\hat v_{\mu\lambda}$ is a bounded self-adjoint operator
acting in $\ell^{2,s}(\Z^2\times\Z^2)$ as
\begin{equation*}\label{two_total}
\hat \H_{\mu\lambda}=\hat \H_0 + \hat \V_{\mu\lambda},\qquad
\mu,\lambda\in\R.
\end{equation*}
Here the free Hamiltonian $\hat \H_0$  of a system of two identical particles (bosons) is a bounded self--adjoint operator acting in $\ell^{2,s}(\Z^2\times\Z^2)$ as
\begin{equation*}
\hat \H_0 \hat f(x,y)= \sum_{n\in\Z^2} \hat \epsilon(x-n) \hat
f(n,y) + \sum\limits_{n\in\Z^2} \hat \epsilon(y-n) \hat f(x,n),
\end{equation*}
where
\begin{equation}\label{def:epsilon}
\hat \epsilon(s) =
\begin{cases}
2 & \text{if $|s|=0,$}\\
-\frac{1}2 & \text{if $|s|=1,$}\\
0 & \text{if $|s|>1,$}
\end{cases} 
\end{equation}
and $|s|=|s_1|+|s_2|$ for $s=(s_1,s_2)\in \Z^2.$

The interaction $\hat \V_{\mu\lambda}$ is the multiplication
operator
\begin{equation*}\label{interaction}
\hat \V_{\mu\lambda} \hat f (x,y) = \hat v_{\mu\lambda}(x-y) \hat
f(x,y),
\end{equation*}
given by the function
\begin{equation}\label{def:potentials}
\hat v_{\mu\lambda}(s)=
\begin{cases}
\mu & \text{if $|s|=0,$}\\
\frac{\lambda}2 & \text{if $|s|=1,$}\\
0 & \text{if $|s|>1.$}
\end{cases} 
\end{equation}

\subsection{The two-particle Hamiltonian: the momentum-space representation}

Let $\T^2= \allowbreak ( \R /2\pi \Z)^2  \equiv [-\pi,\pi)^2$ be the two
dimensional torus, the Pontryagin dual group of $\Z^2$, equipped
with the Haar measure $\d p,$and let $L^{2,s}(\T^2\times\T^2)$ be
the Hilbert space of square-integrable symmetric functions on
$\T^2\times\T^2.$ Let $\cF:\ell^2(\Z^2)\rightarrow L^2(\T^2)$ be the
standard Fourier transform
$$
\cF \hat f(p)=\frac{1}{2\pi} \sum_{x\in\Z^2} \hat f(x) e^{ip\cdot
x},
$$
where $p\cdot x: = p_1x_1+p_2x_2$ for $p=(p_1,p_2)\in\T^2$ and
$x=(x_1,x_2)\in\Z^2.$

In the momentum-space via $\cF\otimes\cF$ the two-particle
Hamiltonian is represented in $L^{2,s}(\T^2\times \T^2)$ as
$$
\H_{\mu\lambda}:=(\cF\otimes\cF) \hat
\H_{\mu\lambda}(\cF\otimes\cF)^*:=\H_0 + \V_{\mu\lambda}.
$$
Here the free Hamiltonian $ \H_0=(\cF \otimes \cF) \hat \H_0
(\cF\otimes \cF)^*$ is the multiplication operator:
$$
\H_0 f(p,q) = [\epsilon(p) + \epsilon(q)]f(p,q),
$$
where
$$
\epsilon(p) := \sum\limits_{i=1}^2 \big(1-\cos p_i),\quad
p=(p_1,p_2)\in \T^2,
$$
is the \emph{dispersion relation} of a single boson.
The interaction  $\V_{\mu\lambda}=(\cF \otimes \cF)\hat\V_{\mu\lambda} (\cF\otimes\cF)^*$ is the (partial) integral operator
$$
\V_{\mu\lambda} f(p,q) = \frac{1}{(2\pi)^2}\int_{\T^2}
v_{\mu\lambda}(p-u) f(u,p+q-u)\d u,
$$
where
$$
v_{\mu\lambda}(p)=\mu+\lambda\sum_{i=1}^2\cos p_i,\quad p=(p_1,p_2)\in \T^2.
$$

\subsection{The Floquet-Bloch decomposition of $\H_{\mu\lambda}$ and discrete Schr\"odinger operators}\label{subsec:von_neuman}

Since $\hat H_{\mu\lambda}$ commutes with the representation of the
discrete group $\Z^2$ by shift operators on the lattice, we can
decompose the space $L^{2,s}(\T^2\times\T^2)$ and $\H_{\mu\lambda}$
into the von Neumann direct integral as
\begin{equation}\label{hilbertfiber}
L^{2,s}(\T^2\times \T^2)\simeq \int\limits_{K\in \T^2} \oplus
L^{2,e}(\T^2)\,\d K
\end{equation}
and
\begin{equation}\label{fiber}
\H_{\mu\lambda} \simeq \int\limits_{K\in \T^2} \oplus
H_{\mu\lambda}(K)\,\d K,
\end{equation}
where $L^{2,e}(\T^2)$ is the Hilbert space of square-integrable even
functions on $\T^2$ (see, e.g., \cite{ALMM:2006_cmp}).

The fiber operator $H_{\mu\lambda}(K),$ $K\in\T^2,$ is a self-adjoint operator  defined in $L^{2,e}(\T^2)$ as
\begin{equation*}
H_{\mu\lambda}(K) := H_0(K) + V_{\mu\lambda},
\end{equation*}
where the unperturbed operator $H_0(K)$ is the multiplication operator by the function
$$
\cE_K(p):= 2 \sum_{i=1}^2\Big(1-\cos\tfrac{K_i}2\,\cos p_i\Big),
$$
and the perturbation  $V_{\mu\lambda}$ is defined as
$$
V_{\mu\lambda} f(p)= \frac{1}{(2\pi)^2}\int_{\T^2}  \Big(\mu+\lambda
\sum\limits_{i=1}^2\cos p_i\cos q_i\Big) f(q)\d q.
$$

In the literature the parameter $K\in\T^2$ is called the
\emph{two-particle quasi-momentum} and the fiber $H_{\mu\lambda}(K)$
is called the \emph{discrete Schr\"odinger operator} associated to
the two-particle Hamiltonian $\hat \H_{\mu\lambda}.$

Using the Fourier transform \eqref{hilbertfiber} and \eqref{fiber}
can be represented as
$$
\ell^{2,s}(\Z^2\times \Z^2)\simeq \int\limits_{K\in \T^2} \oplus
\ell^{2,e}(\Z^2)\,\d K
$$
and
$$
\hat\H_{\mu\lambda} \simeq \int\limits_{K\in \T^2} \oplus \hat
H_{\mu\lambda}(K)\,\d K,
$$
where $\ell^{2,e}(\Z^2)$ is the Hilbert space of square-summable
even functions on $\Z^2$ and
$$
\hat H_{\mu\lambda}(K):=\cF^* H_{\mu\lambda}(K) \cF,
$$
where 
\begin{equation}\label{hash_shapka}
\hat H_{\mu\lambda}(K)=\hat H_0(K) + \hat V_{\mu\lambda}(K), 
\end{equation}
and 
\begin{equation*} 
\hat H_0(K)f(x) \ = \ \sum_{s\in\Z^2}
\hat \cE_K(x-s) \hat f(s),\qquad f\in\ell^{2,e}\Z^2
\end{equation*}
with 
\begin{equation*}
\hat \cE_K(x)=2\hat\epsilon(x)\,\cos\frac{K\cdot x}{2}
\end{equation*}
and the operator $\hat V_{\mu\lambda}$ acts in $\ell^{2,e}(\Z^2)$ as
\begin{equation*}
\hat  V_{\mu\lambda}f(x)= \hat v_{\mu\lambda}(x)f(x),
\end{equation*}%
and the functions $\hat \epsilon$ and $\hat v_{\mu\lambda}$ are given by \eqref{def:epsilon} and  \eqref{def:potentials}, respectively.

\subsection{The essential spectrum of discrete Schr\"odinger operators} \label{subsec:ess_spec}

Since $V_{\mu\lambda}$ has rank at most three, by Weyl's Theorem for
any $K\in\T^2$ the essential spectrum
$\sigma_{\ess}(H_{\mu\lambda}(K))$ coincides with the spectrum of
$H_0(K),$ i.e.,
\begin{equation}\label{eq:essential_spectrum}
\sigma_{\ess}(H_{\mu\lambda}(K))=\sigma(H_0(K)) =
[\cE_{\min}(K),\cE_{\max}(K)],
\end{equation}
where
\begin{align*}
\cE_{\min}(K):= & \min_{p\in  \T ^2}\,\cE_K(p) = 2\sum\limits_{i=1}^2\Big(1-\cos \tfrac{K_{i}}2\Big)\geq 0=\cE_{\min}(0),\\
\cE_{\max}(K):= & \max_{p\in  \T ^2}\,\cE_K(p) =
2\sum\limits_{i=1}^2\Big(1+\cos \tfrac{K_{i}}2\Big)\leq
8=\cE_{\max}(0).
\end{align*}

\section{Main results}\label{sec:main_results}

Our first  main result is the following generalization of \cite[Theorems 1 and
2]{ALMM:2006_cmp}.

\begin{theorem}\label{teo:disc_Kvs0}
Suppose that $H_{\mu\lambda}(0)$ has $n$ eigenvalues below resp.
above the essential spectrum for some $\mu,\lambda\in\R.$ Then for
every $K\in\T^2$ the operator $H_{\mu\lambda}(K)$ has at least $n$
eigenvalues below resp. above its essential  spectrum.
\end{theorem}

Next we find exact lower bound for the number of eigenvalues of
$H_{\mu\lambda}(K)$ depending only on $\mu$ and $\lambda.$

In the $(\mu,\lambda)$-plane let us define the following nine sets:
\begin{equation}\label{nine_sets}
\begin{aligned}
\cS_{01}:= & \Big\{(\mu,\lambda)\in\R^2:\,\,\lambda>\frac{\pi}{4-\pi}\Big\},\\
\cS_{00}:= &\Big\{(\mu,\lambda)\in\R^2:\,\,|\lambda|<\frac{\pi}{4-\pi}\Big\},\\
\cS_{10}:= &\Big\{(\mu,\lambda)\in\R^2:\,\,\lambda<\frac{\pi}{\pi-4}\Big\},\\ \allowbreak 
\cC_0^+:= &\Big\{(\mu,\lambda)\in\R^2:\,\,\lambda\mu - 4\lambda - 2\mu>0,\,\, \lambda<2 \Big\},\\
\cC_1^+:= &\Big\{(\mu,\lambda)\in\R^2:\,\,\lambda\mu - 4\lambda- 2\mu<0\Big\},\\
\cC_2^+:= &\Big\{(\mu,\lambda)\in\R^2:\,\,\lambda\mu- 4\lambda - 2\mu>0,\,\, \lambda>2\Big\}, \\
\cC_0^-:= &\Big\{(\mu,\lambda)\in\R^2:\,\,\lambda\mu+4\lambda+2\mu>0,\,\, \lambda>-2 \Big\},\\
\cC_1^-:= &\Big\{(\mu,\lambda)\in\R^2:\,\,\lambda\mu+4\lambda+2\mu<0\Big\},\\
\cC_2^-:=
&\Big\{(\mu,\lambda)\in\R^2:\,\,\lambda\mu+4\lambda+2\mu>0,\,\,
\lambda<-2\Big\}
\end{aligned}
\end{equation}
(see Figures \ref{fig:antisim}-\ref{fig:sohalar0} below).

\begin{figure}[thp] 
\begin{minipage}{.32\textwidth}
\centering
\includegraphics[width=\textwidth]{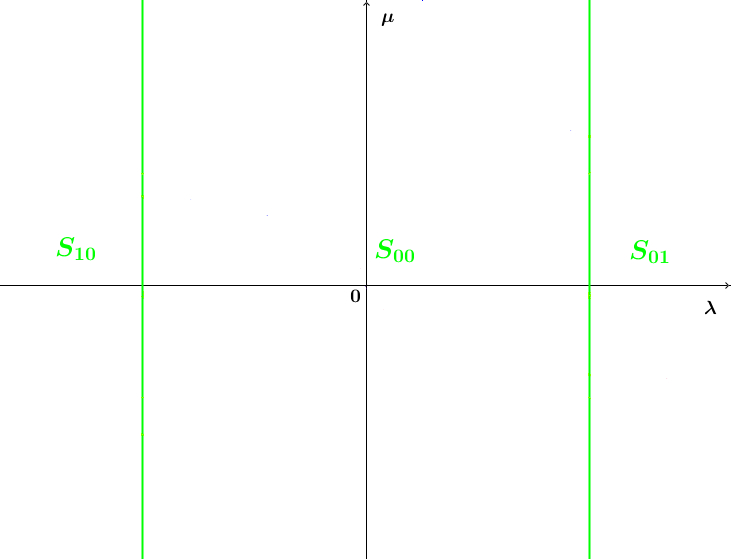}
\subcaption{(a)}\label{fig:antisim}
\end{minipage}
\begin{minipage}{.32\textwidth}
\centering
\includegraphics[width=\textwidth]{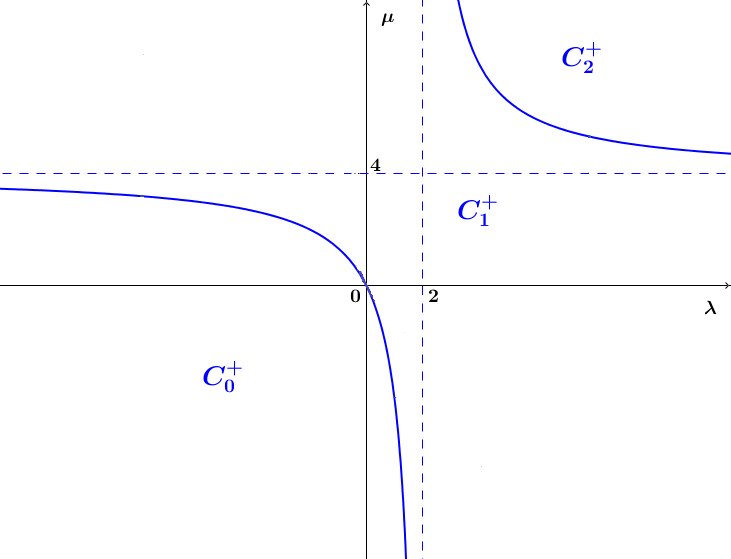}
\subcaption{(b)}
\end{minipage}
\begin{minipage}{.32\textwidth}
\centering
\includegraphics[width=\textwidth]{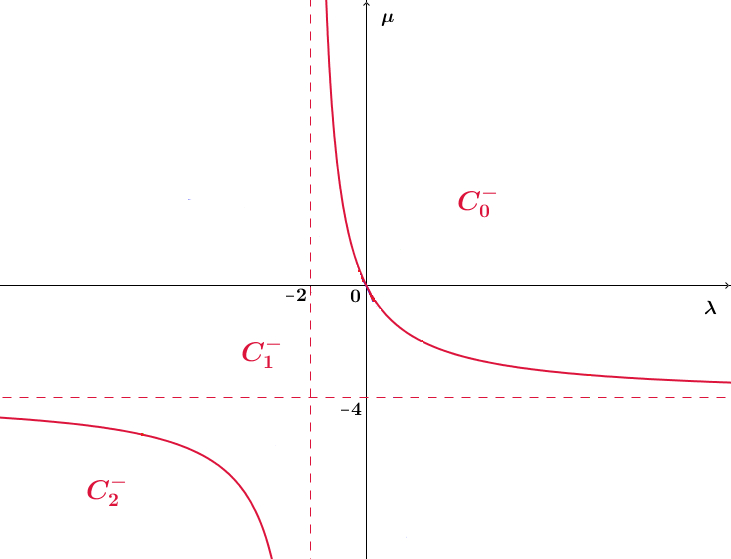}
\subcaption{(c)} 
\end{minipage}
\\
\begin{minipage}{.75\textwidth}
\centering 
\includegraphics[width=\textwidth]{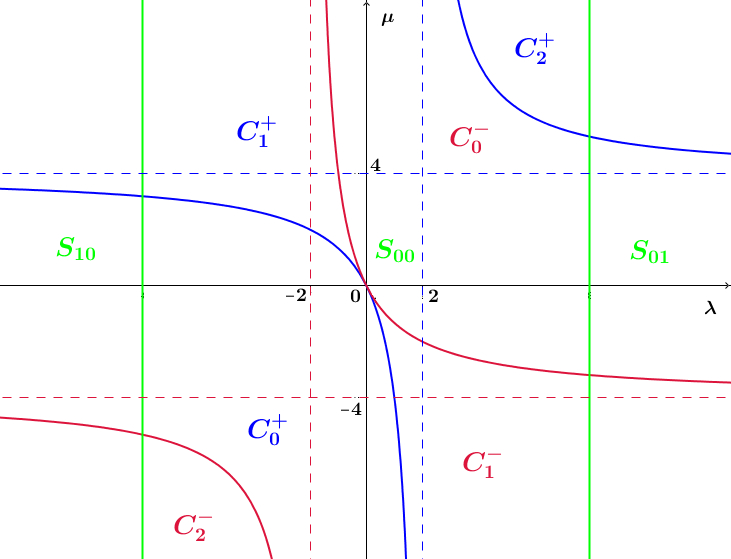} 
\subcaption{(d)}\label{fig:sohalar0}
\end{minipage}
\caption{\small Schematic locations of sets $\cS_{ij}$ and $\cC_i^\pm$ in \eqref{nine_sets}.}
\end{figure}

Let $\miqdor_+(H_{\mu\lambda}(K))$ resp.
$\miqdor_-(H_{\mu\lambda}(K))$ be the number of $H_{\mu\lambda}(K)$
above resp. below its essential spectrum.

\begin{theorem}\label{teo:xosq_kamida}
Let $K\in\T^2$ and $(\mu,\lambda)\in\R^2.$ Then

\begin{equation}\label{above_eigenK}
\begin{aligned}
& (\mu,\lambda)\in \cC_2^+ \cap \cS_{01} & 
\Longrightarrow \qquad  \miqdor_+(H_{\mu\lambda}(K)) =3, \\
& (\mu,\lambda)\in \cC_2^+ \Delta \cS_{01}  & \hspace*{-14mm} \Longrightarrow \qquad  \miqdor_+(H_{\mu\lambda}(K)) \ge2,\\
& (\mu,\lambda)\in \cC_1^+\setminus  \cS_{01} &\hspace*{-14mm}   \Longrightarrow \qquad     \miqdor_+(H_{\mu\lambda}(K)) \ge1,\\
& (\mu,\lambda)\in \overline{\cC_0^+}  &  \hspace*{-8mm}
\Longrightarrow \qquad  \miqdor_+(H_{\mu\lambda}(K)) =0,
\end{aligned}
\end{equation}

and
\begin{equation}\label{below_eigenK}
\begin{aligned}
& (\mu,\lambda)\in \cC_2^- \cap \cS_{10} & 
\Longrightarrow \qquad \miqdor_-(H_{\mu\lambda}(K)) =3, \\
& (\mu,\lambda)\in \cC_2^- \Delta \cS_{10}  & \hspace*{-14mm} \Longrightarrow \qquad \miqdor_-(H_{\mu\lambda}(K)) \ge2,\\
& (\mu,\lambda)\in \cC_1^-\setminus \cS_{10}  & \hspace*{-14mm} \Longrightarrow \qquad \miqdor_-(H_{\mu\lambda}(K)) \ge1,\\
& (\mu,\lambda)\in \overline{\cC_0^-} & \hspace*{-8mm}
\Longrightarrow  \qquad \miqdor_-(H_{\mu\lambda}(K)) =0,
\end{aligned}
\end{equation}
where $A\Delta B:=(A\setminus B)\cup (B\setminus A)$ is the symmetric difference of sets and $\overline{A}$ is the closure of a set $A.$
\end{theorem}

Theorem \ref{teo:xosq_kamida} provides a lower bound for the number
of eigenvalues for both sides of the essential spectrum of
$H_{\mu\lambda}(K)$. Recall that by the min-max principle
$H_{\mu\lambda}(K)$ can have at most three bound states outside its
essential spectrum. 

\begin{remark}
Theorem \ref{teo:xosq_kamida} implies that in some subsets of the $(\mu,\lambda)$-plane (for example in $\cC_1^+\cap \cC_1^-$ ) the eigenvalues of $H_{\mu\lambda}(K)$ can appear simultaneously on both sides of the essential spectrum.   
 \end{remark}

The following result shows that the estimates for
$\miqdor_\pm(H_{\mu\lambda}(K))$ given by Theorem
\ref{teo:xosq_kamida} are sharp.

\begin{theorem}\label{teo:sharpness}
Let $K=0.$ Then all inequalities in \eqref{above_eigenK} and
\eqref{below_eigenK} are in fact equalities. Moreover, the spaces
$$
L^{2,e,s}(\T^2):=\Big\{f\in L^{2,e}(\T^2):\,\, f(p_1,p_2)
=f(p_2,p_1),\,\, p_1,p_2\in\T\Big\}
$$
and
$$
L^{2,e,a}(\T^2):=\Big\{f\in L^{2,e}(\T^2):\,\, f(p_1,p_2)
=-f(p_2,p_1),\,\,p_1,p_2\in\T\Big\}
$$
of symmetric and antisymmetric even functions are invariant with
respect to $H_{\mu\lambda}(0)$ and:
\begin{itemize}
\item[(1)] if $(\mu,\lambda)\in \overline{\cS_{00}},$ then $H_{\mu\lambda}(0)$ has
no antisymmetric bound states outside the essential spectrum;

\item[(2)] if $(\mu,\lambda)\in \cS_{01}$ resp.  $(\mu,\lambda)\in \cS_{10},$ then $H_{\mu\lambda}(0)$
has a unique antisymmetric bound state above resp. below the essential spectrum;

\item[(3)] if $(\mu,\lambda)\in \cC_2^+$ resp.  $(\mu,\lambda)\in \cC_2^-,$  then $H_{\mu\lambda}(0)$ has a exactly two symmetric bound states above resp. below the essential spectrum;

\item[(4)] if $(\mu,\lambda)\in \cC_1^+\cup\partial \cC_2^+$ resp.
$(\mu,\lambda)\in \cC_1^-\cup\partial \cC_2^-,$ then $H_{\mu\lambda}(0)$ has a unique symmetric bound state above resp. below the essential spectrum, where $\partial A$ is the topological
boundary of a set $A;$ 

\item[(5)] if $(\mu,\lambda)\in \overline{\cC_0^+}$ resp.  $(\mu,\lambda)\in \overline{\cC_0^-},$ then $H_{\mu\lambda}(0)$ has no 
symmetric bound states above resp. below the essential spectrum.
\end{itemize}
\end{theorem}

\section{Proofs of the main results}\label{sec:proofs}

\subsection{Discrete spectrum of $H_{\mu\lambda}(0)$} \label{subsec:disc_spec000}

Unlike the case $K\ne0$ in the case $K=0$ the Fredholm determinant
$\Delta_{\mu\lambda}(0,z)$ is easier to study. Notice that the
operator $H_0(0)$ is the multiplication operator by the symmetric
function $\cE_0(p)=2\epsilon(p)$ in $L^{2,e}(\T^2).$ Hence, both
$L^{2,e,s}(\T^2)$ and $L^{2,e,a}(\T^2)$ are invariant with respect
to $H_0(0).$ Moreover, since
\begin{align*}
&2\cos p_1\cos q_1 + 2\cos p_2\cos q_2 \\
= &(\cos p_1 + \cos p_2)(\cos q_1 + \cos q_2) + (\cos p_1 - \cos
p_2)(\cos q_1 - \cos q_2),
\end{align*}
the spaces of symmetric and antisymmetric even functions in $\T^2$
are invariant also with respect $V_{\mu\lambda}.$  Thus,

\begin{lemma}\label{lem:separation_K0}
$$
\sigma(H_{\mu\lambda}(0)) = \sigma(H_{\mu\lambda}^s) \cup
\sigma(H_{\lambda}^a),
$$
where
$$
H_{\mu\lambda}^s: = H_0(0) + V_{\mu\lambda}^s\qquad \text{and}\qquad
H_{\lambda}^a: = H_0(0) + V_{\lambda}^a,
$$
with
$$
V_{\mu\lambda}^sf(p) =\frac{\mu}{4\pi^2} \int_{\T^2} f(q)\,\d q +
\frac{\lambda}{8\pi^2}\,(\cos p_1 + \cos p_2) \int_{\T^2} (\cos q_1
+ \cos q_2)f(q)\,\d q
$$
and
$$
V_{\lambda}^af(p) =\frac{\lambda}{8\pi^2}\,(\cos p_1 - \cos p_2)
\int_{\T^2} (\cos q_1 - \cos q_2)f(q)\,\d q,
$$
are the restrictions of $H_{\mu\lambda}(0)$ onto $L^{2,e,s}(\T^2)$
and $L^{2,e,a}(\T^2).$
\end{lemma}

\begin{theorem}\label{teo:KSexistence}
Fix $(\mu,\lambda)\in\R^2.$

\begin{itemize}
 \item[(a)] If $2\lambda + \mu \ge0,$ then  $H_{\mu\lambda}^s$ has at least one eigenvalue greater than $8.$

\item[(b)] If $2\lambda + \mu \le0,$ then  $H_{\mu\lambda}^s$ has at least one negative eigenvalue.
\end{itemize}

\end{theorem}

\begin{proof}
Note that the trace of $V_{\mu\lambda}^s$ is exactly $2\lambda+\mu.$
Hence, the proof can be done along the essentially same lines of
\cite[Theorem 1.3]{KhLA:2020_arxiv} introducing the non-symmetric
Birman-Schwinger operator in $L^{2,e,s}(\T^2)$ in place of
$L^2(\T^2)$ and rewriting it  as a small perturbation of rank-one projection (which is not identically zero in $L^{2,e,s}(\T^2)$).
\end{proof}

\begin{remark}
A statement in $L^{2,e,a}(\T^2)$ analogous to Theorem
\ref{teo:KSexistence} cannot be proven using the arguments of \cite[Theorem 1.3]{KhLA:2020_arxiv} since in this case the rank-one projection (which is norm-close to the Birman-Schwinger operator) is identically zero in $L^{2,e,a}(\T^2)$.
\end{remark}

Note that the essential spectrum of all Hamiltonians
$H_{\mu\lambda}(0),$ $H_{\mu\lambda}^s$ and $H_{\mu\lambda}^a$
coincide with the segment $[0,8].$ Let us find an (implicit)
equation for the discrete eigenvalues of $H_{\mu\lambda}(0).$ By
Lemma \ref{lem:separation_K0} it is enough to solve
$$
H_{\mu\lambda}^sf=zf \qquad\text{resp.}\qquad  H_{\lambda}^af=zf
$$
in $z\in \C\setminus [0,8]$ and nonzero $f.$  Here we apply
Fredholm's determinants theory (see, e.g.,
\cite{{S:1977_adv.math}}). For shortness, writing
\begin{align*}
& a(z):=  \frac{1}{4\pi^2}\int_{\T^2} \frac{\d p}{\cE_0(p)-z},\qquad
& \hspace*{-6mm} b(z):= \frac{1}{8\pi^2}\int_{\T^2} \frac{(\cos p_1+\cos p_2)^2\,\d p}{\cE_0(p)-z}, \\
& c(z):=  \frac{1}{8\pi^2} \int_{\T^2}\frac{(\cos p_1 + \cos
p_2)\,\d p}{\cE_0(p)-z},\qquad & \hspace*{-6mm} d(z):=
\frac{1}{8\pi^2}\int_{\T^2}\frac{(\cos p_1 - \cos p_2)^2 \,\d
p}{\cE_0(p)-z},
\end{align*}
define the Fredholm determinants associated to $H_{\mu\lambda}^s$
and $H_\lambda^a,$ respectively, as
$$
\Delta_{\mu\lambda}^{(s)}(z) :=  \,
\Delta_{\mu0}(z)\Delta_{0\lambda}(z) - 2\mu\lambda c(z)^2,
$$
and
$$
\Delta_\lambda^{(a)}(z) := \,  1+\lambda d(z),
$$
where
\begin{align*}
&\Delta_{\mu0}(z) := 1+\mu a(z),\qquad \Delta_{0\lambda}(z) := 1 +
\lambda b(z).
\end{align*}

Note that the functions $a(\cdot),$ $b(\cdot)$ and $d(\cdot)$ are
analytic in $\C \backslash[0,8]$, strictly increasing in $\R
\backslash[0,\,8],$ negative in $(8,+\infty)$ and positive in
$(-\infty, 0).$ Their behaviour near $z=0$ and $z=8$ are established
in the following proposition.

\begin{proposition}\label{prop:asymp_determinants}
There exists $\delta>0$ such that for every $\omega
\in\{a,b,c,d\}$ there exist functions $P_\omega^0,$ and
$Q_\omega^0,$  analytic in the disc $\{z\in\C:\,\,|z|<\delta\}$, and
$P_\omega^1,$ and $Q_\omega^1,$  analytic in the disc
$\{z\in\C:\,\,|z-8|<\delta\}$, such that

\begin{itemize}

\item for every $z\in(-\delta,0)$
\begin{equation}\label{expansion_at0}
\omega(z) = P_\omega^0(z)\ln(-z) + Q_\omega^0(z),
\end{equation}
where
$$
P_a^0(0)=-\frac{1}{4\pi},\quad P_b^0(0)=-\frac{1}{2\pi},\quad
P_c^0(0)=-\frac{1}{4\pi},\quad P_d^0(0)=0
$$
and
$$
Q_a^0(0)=\frac{5\ln2}{4\pi},\quad
Q_b^0(0)=\frac{5\ln2-\pi}{2\pi},\quad
Q_c^0(0)=\frac{5\ln2-\pi}{4\pi},\quad Q_d^0(0)=\frac{4-\pi}{\pi};
$$

\item for every $z\in(8,8+\delta)$
\begin{equation}\label{expansion_at8}
\omega(z) = P_\omega^1(z)\ln(z-8) + Q_\omega^1(z),
\end{equation}
where
$$
P_a^1(8)=\frac{1}{4\pi},\quad P_b^1(8)=\frac{1}{2\pi},\quad 
P_c^1(8)=\frac{1}{4\pi}, \quad P_d^1(8)=0
$$
and
$$
Q_a^1(8)=-\frac{5\ln2}{4\pi},\quad
Q_b^1(8)=-\frac{5\ln2-\pi}{2\pi},\quad Q_c^1(8)=-\frac{5\ln2 -
\pi}{4\pi}, \quad Q_d^1(8)=\frac{\pi-4}{\pi}.
$$
\end{itemize}
\end{proposition}

The proof of Proposition \ref{prop:asymp_determinants} is postponed
to the Appendix \ref{sec:append_A}.

\begin{lemma}\label{lem:det_zeros_vs_eigen}
A number $z\in \C\setminus[0,8]$ is an eigenvalue of
$H_{\mu\lambda}^s$ resp. $H_\lambda^a$  of multiplicity $m\ge1$ if
and only if it is a zero of $\Delta_{\mu\lambda}^{(s)}(\cdot)$
resp.$\Delta_\lambda^{(a)}(\cdot)$ of multiplicity $m.$ Moreover, in
$\R\setminus [0,8]$ the function $\Delta_{\mu\lambda}^{(s)}(\cdot)$
has at least one and at most two zeros and the function
$\Delta_\lambda^{(a)}(\cdot)$ has at most one zero.
\end{lemma}

\begin{proof}
The first assertion follows from the  Fredholm determinants theory.
By Theorem \ref{teo:KSexistence} $H_{\mu\lambda}^s$ has at least one
eigenvalue outside the essential spectrum. Moreover, since
$H_{\mu\lambda}^s$ is of rank two resp. rank one, by the min-max principle, it has at most two eigenvalues outside the essential spectrum. Hence, by the first part of the proposition
$\Delta_{\mu\lambda}^{(s)}(\cdot)$ has at least one and at most two
zeros in $\R\setminus[0,8].$

The last assertion follows from the rank-one property of $H_\lambda^{(a)}$ and the first part of the proposition.
\end{proof}

Next we study zeros of $\Delta_\lambda^{(a)}.$

\begin{lemma}\label{lem:antisim_zeros}
Let $\lambda\in\R.$

\begin{itemize}
\item[(a)] If $\lambda\in[\frac{\pi}{\pi-4},\frac{\pi}{4-\pi}],$ then $\Delta_\lambda^{(a)}(\cdot)>0$ in $\R\setminus [0,8].$

\item[(b)] If $\lambda<\frac{\pi}{\pi-4},$ then $\Delta_\lambda^{(a)}(\cdot)$ has a unique zero in $(-\infty,0)$ and $\Delta_\lambda^{(a)}(\cdot)>1$ in $(8,+\infty).$

\item[(c)] If $\lambda>\frac{\pi}{4-\pi},$ then $\Delta_\lambda^{(a)}(\cdot)$ has a unique zero in $(8,\infty)$ and $\Delta_\lambda^{(a)}(\cdot)>1$ in $(-\infty,0).$

\end{itemize}
\end{lemma}

\begin{proof}
Note that the map $z\mapsto d(z)$ strictly increases in both
connected components of $\R\setminus[0,8].$ Hence, the equation
$\Delta_\lambda^{(a)}(z)=0$ has at most one zero in
$\R\setminus[0,8].$ Moreover, 
\begin{equation}\label{behav_anti_infty}
\lim\limits_{z\to\pm\infty} \Delta_\lambda^{(a)}(z)=1
\end{equation} 
and by Proposition \ref{prop:asymp_determinants}
\begin{equation}\label{behav_anti_08}
\lim\limits_{z\nearrow 0} \Delta_\lambda^{(a)}(z)=1 +
\frac{(4-\pi)\lambda}{\pi}, \qquad \lim\limits_{z\searrow 8}
\Delta_\lambda^{(a)}(z)=1 + \frac{(\pi-4)\lambda}{\pi}.
\end{equation}
Now the assertions of the lemma follow from the strict monotonicity
and continuity of $z\mapsto\Delta_\lambda^{(a)}(z),$ and
\eqref{behav_anti_infty}-\eqref{behav_anti_08}.
\end{proof}

The following lemma provides the dependence of the number of zeros
of $\Delta_{\mu\lambda}^{(s)}$ in $(8,+\infty)$ on $\mu$ and
$\lambda.$

\begin{lemma}\label{lem:sim_positive_zeros}
Let $(\mu,\lambda)\in\R^2.$

\begin{itemize}
\item[(a)] If $\lambda\mu-4\lambda -2\mu\ge0$ and $\lambda<2,$ then $\Delta_{\mu\lambda}^{(s)}$ has no zeros in $(8,+\infty).$

\item[(b)] If $\lambda\mu-4\lambda -2\mu<0$ or $\lambda\mu-4\lambda -2\mu=0$ with $\lambda>2,$  then $\Delta_{\mu\lambda}^{(s)}$ has a unique zero in $(8,+\infty).$

\item[(c)] If $\lambda\mu-4\lambda -2\mu>0$ and $\lambda>2,$ then $\Delta_{\mu\lambda}^{(s)}$ has two zeros in $(8,+\infty).$
\end{itemize}

\end{lemma}

\begin{proof}
Note that
\begin{equation}\label{behav_sim_infty}
\lim\limits_{z\to+\infty} \Delta_{\mu\lambda}^{(s)}(z)  =1
\end{equation}
and by Proposition \ref{prop:asymp_determinants}
\begin{align*}
\Delta_{\mu\lambda}^{(s)}(z) =  & -\frac{\lambda\mu-4\lambda -2\mu}{8\pi}\,\ln(z-8) \nonumber \\
& + \Big(1-\frac{5\ln 2}{4\pi}\mu - \frac{5\ln 2 -
\pi}{2\pi}\lambda+\frac{5\ln2 -\pi}{8\pi}\lambda\mu \Big) + o(1)
\end{align*}
as $z\searrow 8$ so that
\begin{equation}\label{behav_sim_8}
\lim\limits_{z\searrow 8} \Delta_{\mu\lambda}^{(s)}(z)  =
\begin{cases}
+\infty & \text{if $\lambda\mu-4\lambda -2\mu>0,$} \\
-\infty & \text{if $\lambda\mu-4\lambda -2\mu<0,$} \\
1 - \frac{\mu}{4} & \text{if $\lambda\mu-4\lambda -2\mu=0.$}
\end{cases}
\end{equation}
Notice that $\lambda\mu-4\lambda -2\mu=0$ and $\lambda>2,$ then
$1-\mu/4<0.$

(a) We observe that if $\lambda\mu-4\lambda -2\mu\ge0$ and
$\lambda<2,$ then either $\lambda\le0$ or $\mu\le0.$ Hence, by the
minmax principle, $H_{\mu\lambda}^s$ can have at most one eigenvalue
above the essential spectrum. By Lemma \ref{lem:det_zeros_vs_eigen}
$\Delta_{\mu\lambda}^{(s)}$ has at most one zero in $(8,+\infty).$
Then \eqref{behav_sim_infty} and \eqref{behav_sim_8} imply that
$\Delta_{\mu\lambda}^{(s)}(z)>0$ in $(8,+\infty),$  otherwise $\Delta_{\mu\lambda}^{(s)}(\cdot)$ would cross
$(8,+\infty)$ at least two times.

(b) If $\lambda\mu-4\lambda -2\mu<0$ or $\lambda\mu-4\lambda
-2\mu=0$ with $\lambda>2,$  then by  
\eqref{behav_sim_infty} 
and \eqref{behav_sim_8} the continuous function
$\Delta_{\mu\lambda}^{(s)}(\cdot)$ changes sign in $(8,+\infty)$ so
that the equation $\Delta_{\mu\lambda}^{(s)}(z)=0$ has at least one
zero in $(8,+\infty).$ If this equation had at least  two zeros,
then $\Delta_{\mu\lambda}^{(s)}(\cdot)$ should have at least three
zeros since $\Delta_{\mu\lambda}^{(s)}(\cdot)$ has different signs
at the endpoints of $(8,+\infty).$ This contradicts to Lemma
\ref{lem:det_zeros_vs_eigen}.

(c) If $\lambda\mu-4\lambda -2\mu>0$ and $\lambda>2,$ then
$\mu+2\lambda>0,$ and hence, by Theorem \ref{teo:KSexistence}
$H_{\mu\lambda}^s$ has at least one eigenvalue above the essential
spectrum. Then by Lemma \ref{lem:det_zeros_vs_eigen}
$\Delta_{\mu\lambda}^{(s)}(\cdot)$ has at least one zero in
$(8,+\infty).$ On the other hand, by \eqref{behav_sim_infty} and
\eqref{behav_sim_8} $\Delta_{\mu\lambda}^{(s)}(\cdot)$ has the same
signs at the endpoints of \blue $(8,+\infty),$ \black hence, by
continuity the equation $\Delta_{\mu\lambda}^{(s)}(z)=0$ has at
least two solutions. Now Lemma \ref{lem:det_zeros_vs_eigen} implies
that $\Delta_{\mu\lambda}^{(s)}(\cdot)$ has two zeros in
$(8,+\infty).$
\end{proof}

The zeros of $\Delta_{\mu\lambda}^{(s)}$ in $(-\infty,0)$ are
studied in the following lemma whose proof can be done along the
lines of Lemma \ref{lem:sim_positive_zeros}.

\begin{lemma}\label{lem:sim_negative_zeros}
Let $(\mu,\lambda)\in\R^2.$

\begin{itemize}
\item[(a)] If $\lambda\mu+4\lambda +2\mu\ge0$ and $\lambda>-2,$ then $\Delta_{\mu\lambda}^{(s)}$ has no zeros in $(-\infty,0).$

\item[(b)] If $\lambda\mu+4\lambda+2\mu<0$ or $\lambda\mu+4\lambda+2\mu=0$ with $\lambda<-2,$  then $\Delta_{\mu\lambda}^{(s)}$ has a unique zero in $(-\infty,0).$

\item[(c)] If $\lambda\mu+4\lambda + 2\mu>0$ and $\lambda<-2,$ then $\Delta_{\mu\lambda}^{(s)}$ has two zeros in $(-\infty,0).$
\end{itemize}

\end{lemma}

\begin{proof}[Proof of Theorem \ref{teo:sharpness}]
The assertions of the theorem follow from Lemmas \ref{lem:separation_K0},
\ref{lem:det_zeros_vs_eigen}, \ref{lem:antisim_zeros},
\ref{lem:sim_positive_zeros} and \ref{lem:sim_negative_zeros}.
%
\end{proof}

\subsection{The discrete spectrum of $H_{\mu\lambda}(K)$}

For every $n\ge1$ define
$$
e_n(K;\mu, \lambda):= \sup\limits_{\phi_1,\ldots,\phi_{n-1}\in
L^{2,e}(\T^2)}\,\,\inf\limits_{\psi
\in[\phi_1,\ldots,\phi_{n-1}]^\perp,\,\|\psi\|=1}
(H_{\mu\lambda}(K)\psi,\psi)
$$
and
$$
E_n(K; \mu,\lambda):= \inf\limits_{\phi_1,\ldots,\phi_{n-1}\in
L^{2,e}(\T^2)}\,\,\sup\limits_{\psi
\in[\phi_1,\ldots,\phi_{n-1}]^\perp,\,\|\psi\|=1}
(H_{\mu\lambda}(K)\psi,\psi).
$$
By the minmax principle, $e_n(K;\mu,\lambda)\le \cE_{\min}(K)$ and
$E_n(K;\mu,\lambda)\ge \cE_{\max}(K).$ Moreover, choosing
$\phi_1\equiv1,$ $\phi_2(p)=\cos p_1$ and $\phi_3(p)=\cos p_2$ we
immediately see that $e_n(K;\mu,\lambda) = \cE_{\min}(K)$ and
$E_n(K;\mu,\lambda) = \cE_{\max}(K)$ for all $n\ge4.$

\begin{lemma}\label{lem:monoton_xos_qiymat}
Let $n\ge1$ and $i\in\{1,2\}.$ Then the map
$$
K_i\in\T \mapsto \cE_{\min}(K) - e_n(K;\mu,\lambda)
$$
is non-increasing in $(-\pi,0]$ and non-decreasing in $[0,\pi]$. Similarly, the map
$$
K_i\in\T \mapsto E_n(K;\mu,\lambda) - \cE_{\max}(K) 
$$
is non-increasing in $(-\pi,0]$ and non-decreasing in $[0,\pi]$.
\end{lemma}

\begin{proof}
Without loss of generality we assume that $i=1.$ Given $\psi\in
L^2(\T^2)$ consider
$$
((H_0(K) - \cE_{\min}(K))\psi,\psi)=\int_{\T^2} \sum\limits_{i=1}^2
\cos\tfrac{K_i}{2}\,\big(1-\cos q_i\big)|\psi(q)|^2\,\d q.
$$
Thus, the map $K_1\in\T\mapsto ((H_0(K) - \cE_{\min}(K))\psi,\psi)$
is non-decreasing in $(-\pi,0]$ and is non-increasing 
in $[0,\pi].$ Since $V_{\mu\lambda}$ is independent of $K,$ from the
definition of $e_n(K;\mu,\lambda)$ the map $K_1\in\T\mapsto
e_n(K;\mu,\lambda) - \cE_{\min}(K)$ is non-decreasing in
$(-\pi,0]$ and is non-increasing  in $[0,\pi].$

The case of $K_i\mapsto E_n(K;\mu,\lambda) - \cE_{\max}(K)$ is similar.
\end{proof}

\begin{proof}[Proof of  Theorem \ref{teo:disc_Kvs0}]
From Lemma \ref{lem:monoton_xos_qiymat} for any $K\in\T^2$ and
$m\ge1$
\begin{equation}\label{eq:min_eigen0}
0\le \cE_{\min}(0) - e_m(0;\mu,\lambda) \le \cE_{\min}(K) -
e_m(K;\mu,\lambda),
\end{equation}
and
$$
E_m(K;\mu,\lambda) - \blue \cE_{\max}(K)\black \ge
E_m(0;\mu,\lambda) - \blue \cE_{\max}(0)\black\ge0.
$$
By assumption, $ e_n(0;\mu,\lambda)$ is a discrete eigenvalue of
$H_{\mu\lambda}(0)$ for some $\mu,\lambda\in\R.$ Thus,
$\cE_{\min}(0) - e_n(0;\mu,\lambda)>0,$ and hence, by
\eqref{eq:min_eigen0} and \eqref{eq:essential_spectrum}
$e_n(K;\mu,\lambda)$ is a discrete eigenvalue of $H_{\mu\lambda}(K)$
for any $K\in\T^2.$ Since $e_1(K;\mu,\lambda)\le \ldots \le
e_n(K;\mu,\lambda)<\cE_{\min}(K),$ it follows that
$H_{\mu\lambda}(K)$ has at least $n$ eigenvalue below the essential
spectrum.

The case of $E_n(K;\mu,\lambda)$ is similar.
\end{proof}

\begin{proof}[Proof of Theorem \ref{teo:xosq_kamida}]
Just combine Theorem \ref{teo:disc_Kvs0} with Theorem
\ref{teo:sharpness}.
\end{proof}

\appendix

\section{Proof of Proposition \ref{prop:asymp_determinants}}
\label{sec:append_A}

The general form of expansions \eqref{expansion_at0} and
\eqref{expansion_at8} of $a,b,c,d$ can be established following, for
example, to \cite{LKhL:2012_TMF}. The first two terms of these
expansions, i.e., the pairs $(P_\omega^0(0),Q_\omega^0(0))$ and
$(P_\omega^1(8),Q_\omega^1(8))$ are given by the following lemma.

\begin{lemma}

\begin{align}
& \int_{\T^2} \frac{\d q_1\d q_2}{2-\cos q_1 - \cos q_2 - z} \nonumber \\
&\hspace*{12mm} =
\begin{cases}
-2\pi\ln(-z) + 8\pi \ln2 + o(1) & \text{as $z\nearrow 0,$}\\
2\pi\ln(z-4) - 8\pi\ln 2 + o(1) & \text{as $z\searrow 4,$}
\end{cases}
\label{asymptot10} \\
& \int_{\T^2} \frac{\cos q_1 \cos q_2\,\d q_1\d q_2}{2-\cos q_1 - \cos q_2 - z} \nonumber \\
& \hspace*{12mm} =
\begin{cases}
-2\pi\ln(-z) + 8\pi(\ln 2-1) + o(1) & \text{as $z\nearrow 0,$}\\
2\pi \ln(z-4) - 8\pi(\ln 2-1) + o(1) & \text{as $z\searrow 4.$}
\end{cases}
\label{asymptot20}\\
& \int_{\T^2} \frac{(\cos q_1 + \cos q_2)\d q_1\d q_2}{2- \cos q_1-\cos q_2-z} \nonumber \\
& \hspace*{12mm} =
\begin{cases}
-4\pi\ln(-z) + 16\pi \ln2 -4\pi^2 + o(1) & \text{as $z\nearrow 0,$}\\
4\pi\ln(z-4) - 16\pi\ln 2 +4\pi^2 + o(1) & \text{as
$z\searrow 4,$}
\end{cases}
\label{asymptot30}\\
& \int_{\T^2} \frac{(\cos q_1 + \cos q_2)^2\d q_1\d q_2}{2- \cos q_1-\cos q_2-z}\nonumber \\
& \hspace*{12mm} =
\begin{cases}
-8\pi\ln(-z) + 32\pi \ln2 - 8\pi^2 + o(1) & \text{as $z\nearrow 0,$}\\
8\pi\ln(z-4) - 32\pi \ln 2 +8\pi^2 + o(1) & \text{as $z\searrow 4,$}
\end{cases}
\label{asymptot40}\\
& \int_{\T^2} \frac{(\cos q_1 - \cos q_2)^2\d q_1\d q_2}{2- \cos
q_1-\cos q_2-z}=
\begin{cases}
32\pi  - 8\pi^2 + o(1) & \text{as $z\nearrow 0,$}\\
- 32\pi + 8\pi^2 + o(1) & \text{as $z\searrow 4.$}
\end{cases}
\label{asymptot50}
\end{align}

\end{lemma}

\begin{proof}
Since the idea is the same, we prove
\eqref{asymptot10}-\eqref{asymptot50} only for $z<0.$  By
simmetricity, it suffices to prove only \eqref{asymptot10} and
\eqref{asymptot20}. First we observe that  if $|u|>1,$ then
\begin{equation}\label{1dim_asymp}
\int_{\T}\frac{\d t}{u - \cos t} =
\begin{cases}
\frac{-2\pi}{\sqrt{u^2 - 1}} & \text{if $u<-1,$}\\
\frac{2\pi}{\sqrt{u^2 - 1}} & \text{if $u>1.$}
\end{cases}
\end{equation}

Let us prove \eqref{asymptot10}.  Since $2-z-\cos q_2>1$ for any
$z\in\R\setminus [0,4]$ and $q_2\in \T,$ by \eqref{1dim_asymp}
$$
\int_{\T^2} \frac{\d q_1\d q_2}{2 - \cos q_1 - \cos q_2 - z}  =
\int_\T \frac{\d q_2}{\sqrt{(2-z-\cos q_2)^2 - 1}}.
$$
Then using the change of variables $q_2:=2\arctan
\big(\frac{z-2}{z-4}\big)^{1/2} v$ we get
$$
\begin{aligned}
\int_{\T^2} \frac{\d q_1\d q_2}{2 - \cos q_1 - \cos q_2 - z}
= & \frac{8\pi}{|2-z|} \int_0^\infty \frac{\d v}{\sqrt{\big(\tfrac{z^2-4z}{(z - 2)^2} + v^2\big)\big(1 + v^2\big)}}\\
= & J_1(z) + J_2(z) + J_3(z),
\end{aligned}
$$
where
$$
\begin{aligned}
J_1(z): = & \frac{8\pi}{|2-z|} \int_0^1\frac{\d v}{\sqrt{\tfrac{z^2-4z}{(z - 2)^2} + v^2}},\\
J_2(z): = & \frac{8\pi}{|2-z|} \int_0^1 \frac{\big((1+v^2)^{-1/2} - 1\big)\,\d v}{\sqrt{\tfrac{z^2-4z}{(z - 2)^2} + v^2}},\\
J_3(z): = &  \frac{8\pi}{|2-z|}  \int_1^\infty \frac{\d
v}{\sqrt{\big(\tfrac{z^2-4z}{(z - 2)^2} + v^2\big)\big(1 +
v^2\big)}}.
\end{aligned}
$$
Note that
$$
J_1(z) = \frac{8\pi}{|2-z|}  \Big(\ln\Big[ 1+
\sqrt{1+\tfrac{z^2-4z}{(z-2)^2}}\Big] - \ln
\tfrac{\sqrt{z^2-4z}}{|z-2|}\Big)=-2\pi\ln(-z) + 4\pi\ln 2 + o(1)
$$
as $z\nearrow0.$ Moreover, $J_2$ and $J_3$ are continuous at $z=0$
and $J_2(0) = 4\pi  \ln\frac{2}{1+\sqrt{2}}$ and $J_3(0) = 4\pi
\ln(1 + \sqrt2).$ Thus, \eqref{asymptot10} follows.

Now we prove \eqref{asymptot20}. Using $\int_\T \cos q_2\d q_2=0$
and \eqref{1dim_asymp}
$$
\begin{aligned}
\int_{\T^2} \frac{\cos q_1\cos q_2\d q_1\d q_2}{2- \cos q_1 - \cos
q_2 -z}  = & \int_{\T} \frac{(2- z - \cos q_2)\cos q_2\,\d
q_2}{\sqrt{(2- z - \cos q_2)^2-1}}.
\end{aligned}
$$
Thus changing variables as $q_2=2\arctan u$ we get
$$
\int_{\T^2} \frac{\cos q_1\cos q_2\d q_1\d q_2}{2- \cos q_1 - \cos
q_2 -z}= I_1(z) + I_2(z),
$$
where
\begin{align*}
I_1(z) := & 8\pi(1-z) \int_0^\infty\frac{\d u}{(1+ u^2)^2\sqrt{-z + (2-z)u^2}\sqrt{2-z+ (4-z)u^2}}\\
I_2(z) := & 8\pi \int_0^\infty\frac{\big(2u^2 - (3-z)u^4\big)\,\d
u}{(1+ u^2)^2\sqrt{-z + (2-z)u^2}\sqrt{2-z+ (4-z)u^2}}.
\end{align*}
Obviously, $I_2$ is continuous at $z=0$ and $I_2(0)=2\pi(\pi - 5).$
We represent $I_1(z)$ as
$$
I_1(z) = I_{11}(z) + I_{12}(z) + I_{13}(z),
$$
where
\begin{align*}
I_{11}(z):=& \frac{8\pi(1-z)}{\sqrt{2-z}}  \int_0^1\frac{\d u}{\sqrt{-z + (2-z)u^2}}\\
I_{12}(z):=& 8\pi(1-z) \int_0^1\frac{\big[(1+ u^2)^{-2}(2-z+ (4-z)u^2)^{-1/2} - (2-z)^{-1/2}\big]\,\d u}{\sqrt{-z + (2-z)u^2}} \\
I_{13}(z):= & 8\pi(1-z) \int_1^\infty\frac{\d u}{(1+ u^2)^2\sqrt{-z
+ (2-z)u^2}\sqrt{2-z+ (4-z)u^2}}.
\end{align*}
Then
$$
\begin{aligned}
I_{11}(z) = & \frac{8\pi(1-z)}{2-z} \Big[\ln \Big( 1 + \sqrt{1 + \frac{z}{z-2}}\Big) - \ln \sqrt{\frac{z}{z-2}}\Big] \\
= &- 2\pi\ln(-z) + 6\pi\ln2 + o(1)
\end{aligned}
$$
as $z\nearrow 0.$ Note that the functions $I_{12}(z)$ and
$I_{13}(z)$ are continuous at $z=0$ and
$$
I_{12}(0)=4\pi \Big[\ln\frac{2}{1+\sqrt3} -\frac{\pi}{6} +
\frac{2-\sqrt3}{4}\Big]
$$
and
$$
I_{13}(0)=4\pi\Big[\ln\frac{\sqrt3+1}{\sqrt2} -\frac{\pi}{3} +
\frac{\sqrt3}{4}\Big].
$$
Hence
$$
I_1(z) = - 2\pi \ln(-z) + 8\pi\ln2 - 2\pi^2 +2\pi + o(1)
$$
and \eqref{asymptot20} follows.
\end{proof}

\subsection*{Acknowledgment}

The first and third author acknowledge support from the Foundation for Basic Research of the Republic of Uzbe\-kistan (Grant No. OT-F4-66). The second author acknowledges support from the Austrian Science Fund (FWF) project M~2571-N32.

\end{document}